\documentclass[12pt,a4paper]{article}

\usepackage{flonat-wp}
\usepackage[harvard]{flonat-bib}

\newcommand{\new}[1]{#1}
\newcommand{\removed}[1]{}

\newcommand{\spadBredMin}{5}
\newcommand{\spadBredMax}{14}

\newcommand{\spadBredNuwMin}{20}
\newcommand{\spadBredNuwMax}{44}

\newcommand{\spadVsWpNuwMin}{13}
\newcommand{\spadVsWpNuwMax}{20}

\title{Differentially Private Auditing Under Strategic Response}
\author{Florian A. D. Burnat\thanks{University of Bath, \texttt{fadb20@bath.ac.uk}.} \and Brittany I. Davidson\thanks{University of Bath, \texttt{bid23@bath.ac.uk}.}}
\date{29 July 2026}

\begin{document}

\begin{titlepage}

    \maketitle

    \begin{abstract}
        \noindent Regulatory audits of \new{artificial intelligence (AI)}\removed{AI} systems may use differential privacy (DP) to protect individuals represented in audit data. We study audit design when a fixed protocol releases privacy-protected aggregate harm statistics and the audited developer strategically responds to the resulting observation channel. We formalize a bilevel Stackelberg game in which an auditor commits to follow-up coverage and DP budget allocation across statistics, and a strategic developer reallocates mitigation in response. We introduce the \emph{welfare-weighted observability gap} $B_w$, an objective that weights residual harm by normalized shortfall in the audit's observation score. An exact comparison identity shows that the effect of strategic response is profile-dependent; in a symmetric two-statistic benchmark, we prove that heterogeneous observability strictly worsens $B_w$ when the less observable statistic receives at least as much shortfall weight. We show by a heterogeneous counterexample that harm-proportional allocation is not universally optimal, derive the auditor's full KKT-coupled hypergradient, and provide a single-level reformulation via the developer's \new{Karush--Kuhn--Tucker (}KKT\new{)} system. We propose Strategic Private Audit Design (SPAD), a projected-hypergradient heuristic for local improvement. We specify a synthetic evaluation protocol across aggregate harm statistics, developer behaviors, and privacy budgets under quadratic mitigation costs and exponential residual-harm reduction; preliminary experiments show that SPAD reduces $B_w$ by \spadBredMin--\spadBredMax\% under uniform welfare and \spadBredNuwMin--\spadBredNuwMax\% under non-uniform welfare across the tested privacy-budget regimes, beating welfare-proportional allocation by \spadVsWpNuwMin--\spadVsWpNuwMax\% in the non-uniform-welfare case.

    \end{abstract}

    \setcounter{page}{0}
    \thispagestyle{empty}

\end{titlepage}

\clearpage

\doublespacing

\section{Introduction}
\label{sec:intro}

\removed{Differential privacy is standard for protecting individuals during audits of machine learning systems.}\new{Privacy protection is an active constraint in the emerging artificial-intelligence (AI) audit-access debate: what auditors may see is negotiated against trade secrets and the privacy of the individuals whose data trained the system \citep{Casper2024-em, Cen2024-ee, Burnat2026-xq}. Deployed differentially private publication systems have also set a precedent for committing to global privacy budgets in high-stakes statistical releases \citep{Abowd2022-mf}. These are motivation and an operational analogue, not evidence that DP-protected AI harm audits are already deployed. Our contribution is conditional and prescriptive: \emph{if} an audit uses a privacy-constrained statistical interface, the allocation of its privacy budget becomes a strategic design problem.} Allocating follow-up capacity and privacy accuracy across aggregate harm statistics creates \emph{uneven observability}: some subgroup metrics or failure categories receive less informative evidence. This asymmetry protects the individuals represented in the audit data but creates strategic incentives for the audited system.

Consider a strategic developer operating the audited system while a fixed audit service releases noisy aggregate statistics, such as subgroup false-positive-rate excesses or safety-failure rates for named test categories. We use \emph{developer} for this strategic agent, whether an AI platform, model provider, or in-house team. If the developer anticipates that some statistics will receive less follow-up or privacy accuracy, it has an incentive to reallocate mitigation away from those less observable categories and concentrate residual harm there. The privacy mechanism, meant to protect individuals in the audit data, can therefore create an opportunity to shift failures toward dimensions carrying less audit exposure. We ask: if the auditor anticipates this response, how should it allocate follow-up capacity and privacy budgets to reduce the observability gap?

Strategic responses to evaluations have been studied in classifier games and strategic classification \citep{Hardt2015-ci, Dong2018-ht}, where agents manipulate features after learning how they will be evaluated. Audit games model bilevel competition between an auditor and an adversary that reallocates defensive resources \citep{Blocki2015-du}. Empirical privacy auditing, meanwhile, tests whether trained models satisfy claimed differential privacy (DP) guarantees \citep{Tramer2022-zi, Nasr2023-if}. Our setting combines three different elements: (a) DP composition across aggregate harm statistics; (b) the auditor's commitment to fixed follow-up and privacy allocations that the developer observes; and (c) the developer's best-response reallocation of mitigation effort. \removed{Privacy-auditing papers design audits without modeling the strategic response to the audit interface itself.}\new{The DP-auditing literature performs a verification task orthogonal to ours; work on privacy-protected audit and disclosure interfaces, in turn, does not model the audited party's strategic response to the observation channel.} Strategic-classification papers assume fully observed features, while audit-game papers omit DP constraints. The intersection of privacy-constrained aggregate observation, auditor commitment, and strategic mitigation is the setting studied here.

We formalize this gap as a Stackelberg game between an auditor and a developer. The fixed workload releases all $d$ aggregate statistics. The auditor commits to a follow-up-capacity share $\pi$ and a privacy-budget allocation $\varepsilon=(\varepsilon_1,\ldots,\varepsilon_d)$, subject to a global privacy-loss envelope and basic composition. A strategically rational developer anticipates this commitment and reallocates mitigation to minimize an audit-exposure score plus mitigation cost, thereby concentrating residual harm where the effective observability score $\eta_j:=\pi_j\alpha_j(\varepsilon_j)$ is lowest. Here $h_j$ is baseline harm, $w_j$ its welfare weight, $g_j(h_j,m_j)$ residual harm after mitigation $m_j$, and $\alpha_j(\varepsilon_j)$ normalized informativeness above test size. We introduce the \emph{welfare-weighted observability gap},
\[
B_w(\pi,\varepsilon) = \sum_j w_j\, \bigl(1 - \pi_j\alpha_j(\varepsilon_j)\bigr)\, g_j(h_j, m^*_j(\pi,\varepsilon)).
\]
This objective weights true residual harm by normalized observability shortfall. Its non-negativity follows algebraically from non-negative welfare weights and residual harms and from $0\le\pi_j\alpha_j\le1$, not from DP itself. Because $\pi_j\alpha_j$ is a reduced-form score rather than a probability, $B_w$ is not literally expected undetected harm. Strategic reallocation may increase or decrease it relative to a non-strategic baseline; we derive an exact comparison identity and prove strict worsening in a symmetric two-statistic benchmark.

Our contributions are threefold. \emph{First}, we formalize privacy-constrained auditing as a Stackelberg game over per-statistic follow-up and privacy-budget allocation against a strategic developer, and define the welfare-weighted observability gap $B_w$ as the central audit-design objective. \emph{Second}, we show that strategic worsening is profile-dependent, prove it under transparent sufficient conditions in a symmetric two-statistic benchmark, and give a heterogeneous counterexample showing that harm-proportional allocation is not universally optimal. \emph{Third}, we derive a single-level reformulation and the full KKT-coupled hypergradient, then propose Strategic Private Audit Design (SPAD) as a projected-hypergradient heuristic for local improvement; we specify a synthetic evaluation protocol and report a preliminary ablation in \S\ref{sec:ablation}, with broader empirical validation deferred to a companion paper.

\S\ref{sec:model}--\ref{sec:experiments} formalize the model, prove the strategic blind-spot result, present SPAD, and report a preliminary ablation; \S\ref{sec:related} positions the contribution; and \S\ref{sec:limitations}--\ref{sec:conclusion} state the model's limits and conclusions.

%====================================================================

\section{Model}
\label{sec:model}

We formalize privacy-constrained auditing as a bilevel game in which an auditor allocates follow-up capacity and privacy budgets across aggregate harm statistics, anticipating a strategic developer who reallocates mitigation efforts in response to the resulting audit interface.

\subsection{Setup: harm space and welfare}
\label{sec:setup}

We model a system that can cause harm across $d$ dimensions. Let $j \in \{1, \ldots, d\}$ index these dimensions, which are aggregate statistics or query groups---for example, subgroup--metric pairs, robustness scenarios, or named safety-failure categories---rather than coordinates of an individual feature vector. Let $h \in \mathbb{R}_+^d$ denote the baseline population-harm vector, where $h_j$ is the latent harm in dimension $j$ in the absence of any mitigation. Let $w \in \mathbb{R}_+^d$ be a welfare weight vector, where $w_j$ reflects the social importance of reducing harm in dimension $j$.

\new{\emph{Running example (group-fairness audit).} A regulator audits a deployed screening model using a protected dataset $D$ containing predictions, outcomes, and protected attributes. Dimension $j$ is a subgroup--metric pair, and $q_j(D)$ is a clipped aggregate such as that subgroup's false-positive-rate excess. The baseline population harm is $h_j$; engineering effort $m_j$ reduces it to $g_j(h_j,m_j)$; and the fixed audit service releases a differentially private (DP) estimate $y_j=M_{j,\varepsilon_j}(D)$ of $q_j(D)$. DP protects the individuals represented in $D$. The auditor allocates follow-up capacity and privacy accuracy across these statistics, while the developer observes that allocation and shifts mitigation toward statistics with greater audit exposure. The 2020 US Census TopDown system is an operational analogue---not an AI audit---because it allocates a fixed global privacy-loss budget across geographic levels and measurement groups \citep{Abowd2022-mf}. Our question is how such an allocation should change when the audited party reacts to it.}

The auditor's welfare-relevant target is $\sum_{j=1}^d w_jg_j(h_j,m_j)$, the sum of welfare-weighted residual harms after mitigation. Vector $w$ represents the priorities of an external auditor or regulator, not the developer's objective. We treat the auditor as a \emph{normative} welfare planner: the framework asks how a regulator that cares about welfare-weighted harm should design audits, not how any particular regulator currently behaves; the latter is a positive question left to follow-up work and discussed as a misuse pathway in Appendix~\ref{sec:broader-impacts}.

\subsection{Auditor: follow-up and privacy-budget allocation}
\label{sec:auditor-action}

\removed{The auditor's action consists of two choices: a query policy $\pi \in \Delta^d$ specifying the distribution of audit queries over the $d$ harm dimensions ($\pi_j$ is the probability of querying dimension $j$), and a privacy budget allocation $\varepsilon = (\varepsilon_1, \ldots, \varepsilon_d) \in \mathbb{R}_+^d$, where $\varepsilon_j$ is the DP budget for queries to dimension $j$.}\new{The fixed workload releases every aggregate statistic. The auditor chooses a follow-up allocation $\pi\in\Delta^d$, where $\pi_j$ is the normalized share of finite regulatory investigation or enforcement capacity committed to evidence from statistic $j$ after release, and a privacy-budget allocation $\varepsilon=(\varepsilon_1,\ldots,\varepsilon_d)\in\mathbb{R}_+^d$, where $\varepsilon_j$ is the cumulative privacy budget allocated to statistic $j$'s releases. Thus $\pi_j$ is neither a query probability nor an additional privacy expenditure.} These choices are subject to
\[
\sum_{j=1}^d \varepsilon_j \le \varepsilon_{\mathrm{tot}}, \quad \pi \in \Delta^d, \quad \varepsilon_j \ge 0.
\]
\removed{For each query to dimension $j$, the audit interface returns a noisy observation $y_j = M_j(h^{\mathrm{res}}_j; \varepsilon_j)$, where $M_j$ is a differentially private mechanism, and whose accuracy improves with $\varepsilon_j$.}\new{Let $D$ be the protected audit dataset and $q_j(D)$ the aggregate harm statistic for dimension $j$. The latent quantity being estimated is the population residual harm $g_j(h_j,m_j)$, while the realized release is $y_j=M_{j,\varepsilon_j}(D)$, a noisy estimate of $q_j(D)$ produced by the protocol-fixed mechanism family. The global envelope $\varepsilon_{\mathrm{tot}}$ is exogenous to the allocation problem: it is fixed by the data controller, regulator, or audit agreement, and the auditor distributes it across the workload. Endogenizing this envelope would add a privacy--detection welfare problem and is outside the present game.}

\new{\begin{definition}[Differential privacy of the audit interface]\label{def:dp}
For adjacent datasets $D\sim D'$ differing in one protected audit record, every measurable output set $S$, and each statistic $j$, the protocol-fixed release satisfies
\[
\Pr[M_{j,\varepsilon_j}(D)\in S] \le e^{\varepsilon_j}\Pr[M_{j,\varepsilon_j}(D')\in S]+\delta^{\mathrm{priv}}_j.
\]
The paper optimizes $\varepsilon_j$ while $\delta^{\mathrm{priv}}_j$ is fixed, and $\delta^{\mathrm{priv}}_j=0$ for the pure-DP instantiations. Under our basic-composition abstraction for potentially overlapping statistics, $\sum_j\varepsilon_j\le\varepsilon_{\mathrm{tot}}$, with the $\delta^{\mathrm{priv}}_j$ values composing separately under approximate DP. Parallel composition for genuinely disjoint groups and privacy amplification from subsampling yield different feasible sets and are outside the present model. Appendix~\ref{app:dp-instantiations} instantiates the calibrated Gaussian, Laplace, and randomized-response mechanisms.
\end{definition}}

\new{\begin{remark}[The mechanism is protocol-mandated, not developer-chosen]\label{rem:mechanism-ownership}
For each dimension $j$, the regulator or a trusted audit service fixes and executes---or verifiably enforces---a calibrated family $\{M_{j,\varepsilon}\}_{\varepsilon\ge0}$. Choosing $\varepsilon_j$ selects the protocol-specified member and its associated privacy guarantee and normalized power curve. The developer observes the design but can choose only mitigation $m$; it cannot replace, post-process, or add noise to the audit channel. If the developer controls an unverifiable release mechanism, or if the feasible set contains every mechanism satisfying only an upper-bound DP guarantee, then a constant or deliberately under-calibrated mechanism is admissible and the objection applies. That expanded threat model is not covered by SPAD.
\end{remark}}

\new{The modeled mechanism class is deliberately narrow: calibrated privacy-protected releases of scalar aggregate harm statistics followed by a detection rule. Assumption~\ref{ass:detectability} is not a property of arbitrary DP algorithms; DP stochastic gradient descent and DP $k$-means produce trained models or clusterings rather than answer this audit-statistic workload and therefore lie outside the model. Within the modeled class, we abstract power above test size via a normalized informativeness function.}

\begin{assumption}
\label{ass:detectability}
For each dimension $j$ there exists a normalized informativeness function $\alpha_j : \mathbb{R}_+ \to [0,1]$ that is strictly increasing and continuously differentiable, with $\alpha_j(0) = 0$ and $\lim_{\varepsilon \to \infty} \alpha_j(\varepsilon) = 1$. If the fixed release is followed by a level-$a$ test with power $\widetilde\alpha_j$, then $\alpha_j=(\widetilde\alpha_j-a)/(1-a)$ at the reference operating point. Hence $\alpha_j$ measures power above size on a normalized scale; it is not itself the probability that the audit detects harm.
\end{assumption}

\new{The realized $y_j$ is a random mechanism output; $\alpha_j$ integrates over that noise and therefore is not a function of a particular realized $y_j$. Canonical mechanisms induce the fuller power function $\widetilde\alpha_j(\varepsilon_j,g_j(h_j,m_j))$. We freeze its second argument at a reference operating point $g_j(h_j,m_j^0)$ and normalize power above size to obtain the tractable reduced form $\alpha_j(\varepsilon_j)$. Appendix~\ref{app:dp-instantiations} derives this construction for Gaussian, Laplace, and randomized-response mechanisms. Suppressing the residual-harm argument is a modeling approximation rather than an identity and limits the external interpretation of the score.}

\subsection{Developer: mitigation under a cost budget}
\label{sec:developer-action}

\new{The baseline harm vector $h$ is a state of the system, not a developer action. The developer knows its own baseline profile and mitigation technology and chooses only a mitigation allocation $m\in\mathbb{R}_+^d$, where $m_j$ is engineering effort directed at reducing harm in dimension $j$, subject to a total cost budget $B$:}
\[
m_j \ge 0 \text{ for all } j, \qquad C(m) = \sum_{j=1}^d c_j(m_j) \le B,
\]
with $c_j : \mathbb{R}_+ \to \mathbb{R}_+$ the cost function for mitigation in dimension $j$. Residual harm is determined jointly by the state and the action as $g_j(h_j,m_j)$, with $g_j$ decreasing in its second argument and $g_j(h_j,0)=h_j$. We use $g_j(h_j,m_j)$ in equations and formal statements and reserve ``residual harm'' for prose. Standard examples are $g_j(h_j,m_j)=h_j\exp(-\beta_jm_j)$ and $g_j(h_j,m_j)=\max\{h_j-\beta_jm_j,0\}$.

\begin{assumption}
\label{ass:cost-mitigation}
Each cost function $c_j$ is strictly increasing, strictly convex, and coercive ($c_j(m_j) \to \infty$ as $m_j \to \infty$) with $c_j(0) = 0$. The residual harm function $g_j(h_j,\cdot)$ is strictly convex in $m_j$ for each fixed $h_j > 0$.
\end{assumption}

The coercivity of $c_j$ rules out cost functions with bounded growth (it holds for power-law $c_j(m_j) = m_j^p / p$ with $p > 1$ and is needed for compact sublevel sets in Lemma~\ref{lem:uniqueness}). The strict convexity of $g_j(h_j, \cdot)$ implies a diminishing marginal effectiveness of the mitigation effort, and the strict convexity of $c_j$ captures an increasing marginal cost.

\subsection{Developer's objective: minimize audit exposure plus cost}
\label{sec:developer-objective}

The central asymmetry is that the developer does not minimize welfare-weighted residual harm; it minimizes an audit-exposure score plus cost. The exposure contribution in dimension $j$ is
$\mathrm{DetHarm}_j(\pi, \varepsilon, m) = \pi_j \alpha_j(\varepsilon_j) g_j(h_j,m_j)$.
The developer solves
\begin{equation}
\label{eq:dev-problem}
m^*(\pi, \varepsilon) \in \arg\min_{m \in \mathbb{R}_+^d : C(m) \le B} \sum_{j=1}^d \pi_j \alpha_j(\varepsilon_j) g_j(h_j, m_j) + C(m).
\end{equation}
This best response minimizes the developer's normalized audit-exposure score plus mitigation cost, not welfare-weighted true residual harm. This misalignment can generate strategic blind spots.

\subsection{Auditor's objective: bilevel optimization}
\label{sec:auditor-objective}

The auditor anticipates the developer's response and solves the bilevel problem
\begin{equation}
\label{eq:auditor-problem}
\min_{\pi \in \Delta^d, \, \varepsilon \ge 0 : \sum_j \varepsilon_j \le \varepsilon_{\mathrm{tot}}} B_w(\pi, \varepsilon),
\quad B_w := \sum_j w_j (1 - \pi_j \alpha_j(\varepsilon_j))\, g_j(h_j, m^*(\pi, \varepsilon)_j),
\end{equation}
the welfare-weighted observability gap (Definition~\ref{def:blind-spot}). We retain $\mathrm{TRH}(\pi, \varepsilon) := \sum_j w_j g_j(h_j, m^*_j)$ as a diagnostic; $B_w \le \mathrm{TRH}$ pointwise. The Stackelberg structure parallels strategic classification \citep{Hardt2015-ci, Dong2018-ht} and security games \citep{Blocki2015-du, Korzhyk2011-qp, Schlenker2017-le}.

\subsection{Audit exposure and the observability gap}
\label{sec:blind-spot-gap}

\new{Three quantities organize everything that follows. \emph{Audit exposure} is residual harm weighted by follow-up share and normalized informativeness. \emph{True residual harm} is what society bears at the developer's best response, weighted by welfare. The \emph{observability gap} $B_w$ is residual harm weighted by the normalized observability shortfall $1-\pi_j\alpha_j(\varepsilon_j)$ and is the auditor's design objective throughout. Neither $\pi_j\alpha_j$ nor its complement is a literal event probability.}

\begin{definition}[Audit-exposure score]
\label{def:detected-harm}
$\mathrm{DH}(\pi, \varepsilon, m) = \sum_{j=1}^d \pi_j \alpha_j(\varepsilon_j) g_j(h_j, m_j).$
\end{definition}

\begin{definition}[True residual harm]
\label{def:true-residual-harm}
$\mathrm{TRH}(\pi, \varepsilon) = \sum_{j=1}^d w_j g_j\bigl(h_j, m^*(\pi, \varepsilon)_j\bigr).$
\end{definition}

\begin{definition}[Welfare-weighted observability gap]
\label{def:blind-spot}
\[
B_w(\pi, \varepsilon) \;=\; \sum_{j=1}^d w_j\, \bigl(1 - \pi_j \alpha_j(\varepsilon_j)\bigr)\, g_j\bigl(h_j, m^*(\pi, \varepsilon)_j\bigr).
\]
\end{definition}

$B_w$ weights true residual harm by the normalized observability shortfall $1-\pi_j\alpha_j(\varepsilon_j)$. \removed{$B_w \ge 0$ under DP.}\new{Its non-negativity is algebraic, following from $w_j,g_j\ge0$ and $0\le\pi_j\alpha_j(\varepsilon_j)\le1$; it is not a consequence of DP itself.} Because the weight is a reduced-form score rather than a miss probability, $B_w$ should be interpreted as an audit-design objective, not as the expected amount of literally undetected harm. Theorem~\ref{thm:naive-blindspot} gives a sufficient benchmark in which strategic reallocation strictly increases this objective; outside that benchmark its sign relative to a comparator is profile-dependent.

\subsection{A worked example with three dimensions}
\label{sec:example}

Consider $d = 3$ harm dimensions with $h = (1,1,1)$, heterogeneous welfare weights $w = (3,1,1)$, $\varepsilon_{\mathrm{tot}} = 3$, $\alpha_j(\varepsilon_j) = 1 - \exp(-\varepsilon_j)$, smooth residual harm $g_j(h_j, m_j) = h_j \exp(-m_j)$, $c_j(m_j) = m_j^2/2$, and developer cost budget $B = 1.5$. The numerics were computed using an exact KKT solver (script in the replication package).

Under the na\"ive uniform allocation ($\pi_j = 1/3$, $\varepsilon_j = 1$), $\alpha_j(1) \approx 0.632$ on every dimension; the developer's interior best response is $m^*_j \approx 0.177$ uniformly, with $g_j(h_j,m_j^*) \approx 0.838$, $\mathrm{DH} \approx 0.530$, $\mathrm{TRH} \approx 4.191$, and $B_w \approx 3.308$. Under a welfare-aware policy $\pi = (0.6, 0.2, 0.2)$, $\varepsilon = (2.4, 0.3, 0.3)$, the normalized informativeness scores become $\alpha = (0.909, 0.259, 0.259)$. The strategic developer concentrates mitigation on the more observable statistic ($m^* \approx (0.375, 0.049, 0.049)$), so $g(h,m^*) \approx (0.687, 0.952, 0.952)$: residual harm shifts toward less observable statistics, but those have small welfare weights; therefore, $\mathrm{TRH}$ drops to $\approx 3.966$ and $B_w$ to $\approx 2.742$ ($\sim$17\% reduction). In the running audit, the entries are subgroup false-positive-rate excesses estimated by the releases $M_{j,\varepsilon_j}(D)$; the welfare-aware auditor spends scarce follow-up capacity and privacy budget where society cares most, anticipates that the developer will respond to observability, and accepts higher residual harm only in the lower-weighted subgroups.

Table~\ref{tab:notation} summarizes the core notation used throughout.

\begin{table}[t]
\centering
\small
\caption{Core notation.}
\label{tab:notation}
\begin{tabular}{@{}ll@{}}
\toprule
Symbol & Meaning \\
\midrule
$h_j,w_j$ & Baseline population harm and welfare weight for statistic $j$ \\
$m_j,c_j,B$ & Mitigation effort, its cost, and the developer's cost budget \\
$g_j(h_j,m_j)$ & Residual harm after mitigation \\
$\pi_j,\varepsilon_j$ & Follow-up-capacity share and cumulative privacy budget \\
$\delta_j^{\mathrm{priv}}$ & Fixed approximate-DP failure probability \\
$\alpha_j(\varepsilon_j)$ & Normalized informativeness at the reference operating point \\
$\eta_j:=\pi_j\alpha_j(\varepsilon_j)$ & Effective observability score \\
$m^*(\pi,\varepsilon)$ & Developer's best response \\
$\mathrm{DH},\mathrm{TRH},B_w$ & Audit exposure, true residual harm, and observability gap \\
$\tau,\gamma$ & SPAD convergence tolerance and step size \\
\bottomrule
\end{tabular}
\end{table}

%====================================================================

\section{Strategic Blind Spots Under Privacy-Constrained Auditing}
\label{sec:blind-spots}

This section isolates the strategic mechanism and states a sufficient result. Privacy-constrained releases can create uneven observability across aggregate statistics, which a strategically rational developer may exploit by shifting mitigation away from less observable categories. The effect on the auditor's objective is not universal: it depends on the comparator and on how welfare weights align with the resulting residual-harm shift.

\subsection{Developer best response under fixed audit policy}

Fix an audit policy $(\pi, \varepsilon)$. Given the lower-level problem~\eqref{eq:dev-problem}, the developer's objective is convex on a convex feasible region. We have:

\begin{lemma}[Existence and uniqueness of best response]
\label{lem:uniqueness}
Suppose Assumptions~\ref{ass:detectability}--\ref{ass:cost-mitigation} hold and that $c_j$ are strictly convex and $g_j(h_j, \cdot)$ strictly convex in $m_j$. For any fixed $(\pi, \varepsilon)$ with $\pi \in \Delta^d$ and $\varepsilon \ge 0$, the developer's lower-level problem admits a unique best-response mitigation allocation $m^*(\pi, \varepsilon)$.
\end{lemma}

\begin{proof}[Proof sketch]
The objective in~\eqref{eq:dev-problem} is a sum of strictly convex per-dimension terms ($\pi_j \alpha_j(\varepsilon_j) g_j(h_j, \cdot) + c_j(\cdot)$) over the convex feasible set $\{m \ge 0 : C(m) \le B\}$. Strict convexity provides at most one minimizer. Existence follows because $c_j(m_j) \to \infty$ as $m_j \to \infty$ (strict convexity of $c_j$ with $c_j(0) = 0$ implies coercivity); thus, the sublevel sets of the objective are compact in the feasible set. Hence, $m^*(\pi, \varepsilon)$ exists and is unique.
\end{proof}

At an interior optimum where $m^*_j > 0$, the first-order condition for dimension $j$ is
\begin{equation}
\pi_j \alpha_j(\varepsilon_j) \left| \frac{\partial g_j}{\partial m_j} \right| = (1 + \lambda)\, c'_j(m_j),
\label{eq:foc-dev}
\end{equation}
where $\lambda \ge 0$ is the multiplier on the cost budget constraint $C(m) \le B$. The marginal cost of mitigation in dimension $j$ has two components: the soft cost $c'_j(m_j)$ from the additive $C(m)$ term in the developer's objective and the shadow price $\lambda c'_j(m_j)$ of the hard budget; together, they give $(1+\lambda)\,c'_j(m_j)$. The developer equalizes the marginal reduction in observability-weighted exposure per unit (soft + shadow) cost across active dimensions; dimensions with high $\pi_j \alpha_j(\varepsilon_j)$ require more mitigation effort.

\subsection{The mechanism of blind spots}

Consider two otherwise symmetric dimensions, $j,k$, with different effective observability, $\eta_j:=\pi_j\alpha_j(\varepsilon_j)>\eta_k:=\pi_k\alpha_k(\varepsilon_k)$. By~\eqref{eq:foc-dev}, the developer chooses $m^*_j>m^*_k$, allocating more mitigation to the more observable statistic. Hence $g_k(h_k,m^*_k)>g_j(h_j,m^*_j)$: the less observable statistic retains more residual harm.

The audit-exposure score $\mathrm{DH}(\pi,\varepsilon,m^*)=\sum_\ell\eta_\ell g_\ell(h_\ell,m^*_\ell)$ therefore places a smaller normalized observation weight on the dimension with more residual harm. In contrast, $\mathrm{TRH}(\pi,\varepsilon)=\sum_\ell w_\ell g_\ell(h_\ell,m^*_\ell)$ weights residual harm by welfare. Misalignment can produce a blind spot.

\subsection{A sufficient benchmark for strategic worsening}

For any fixed audit policy and any comparator $m^{\mathrm{ns}}$, writing $r_j^*:=g_j(h_j,m_j^*)$, $r_j^{\mathrm{ns}}:=g_j(h_j,m_j^{\mathrm{ns}})$, and $u_j:=w_j(1-\eta_j)$ gives the exact identity
\begin{equation}
B_w(m^*)-B_w(m^{\mathrm{ns}})=\sum_j u_j\bigl(r_j^*-r_j^{\mathrm{ns}}\bigr).
\label{eq:gap-decomp}
\end{equation}
Anti-monotonicity of $r^*$ and $\eta$ does not by itself determine the sign of this sum. The following symmetric benchmark supplies sufficient conditions under which strategic response is strictly worse than an equal, non-strategic allocation.

\begin{theorem}[Strategic worsening in a symmetric two-statistic audit]
\label{thm:naive-blindspot}
Fix $d=2$, $h_1=h_2=h>0$, $g_j(h,m)=he^{-m}$, and $c_j(m)=m^2/2$. Let a fixed audit policy induce $0<\eta_2<\eta_1<1$, and suppose the unique interior strategic best response exhausts the hard budget $B$. Define $u_j:=w_j(1-\eta_j)$ and assume $u_2\ge u_1>0$. For the symmetric non-strategic comparator $m^{\mathrm{ns}}=(\sqrt{B},\sqrt{B})$,
\[
B_w(m^*)>B_w(m^{\mathrm{ns}}).
\]
\end{theorem}

\begin{proof}[Proof sketch]
The first-order conditions imply $\eta_ihe^{-m_i^*}=(1+\lambda)m_i^*$. Since $m\mapsto me^m$ is strictly increasing, $\eta_1>\eta_2$ implies $m_1^*>m_2^*$. Budget exhaustion gives $(m_1^*)^2+(m_2^*)^2=2B$, so $m_1^*>\sqrt B>m_2^*$. Convexity of $e^{-m}$ and the root-mean-square inequality imply $e^{-m_1^*}+e^{-m_2^*}>2e^{-\sqrt B}$. Applying identity~\eqref{eq:gap-decomp} and $u_2\ge u_1$ then yields the strict inequality. Appendix~\ref{app:proofs} gives the complete argument.
\end{proof}

\subsection{What drives the gap: observability, cost, and welfare}

Three factors shape the strategic gap. \emph{Observability heterogeneity} changes the developer's mitigation incentives. \emph{Cost curvature} governs how strongly mitigation reallocates. \emph{Welfare alignment} determines whether that reallocation raises or lowers $B_w$ relative to the chosen comparator. Equation~\eqref{eq:gap-decomp} makes clear why no sign follows from heterogeneity alone.

\subsection{Corollary: harm-proportional allocation need not be optimal}

\new{The natural regulatory instinct---audit hardest where baseline harm is largest---can fail because baseline harm alone omits welfare, detection elasticity, and the developer's response. The next result is an existence claim, not a universal rejection of harm-proportionality:}

\begin{corollary}[Harm-proportional allocation is not universally optimal]
\label{cor:harm-prop}
There exist heterogeneous parameter profiles satisfying Assumptions~\ref{ass:detectability}--\ref{ass:cost-mitigation} and admitting an interior best response for which the harm-proportional policy $(\pi^{\mathrm{hp}}_j,\varepsilon^{\mathrm{hp}}_j)\propto h_j$ does not minimize $B_w(\pi,\varepsilon)$. In particular, let $h=(1,1)$, $w=(1,1)$, $\varepsilon_{\mathrm{tot}}=2$, $B=2$, $\alpha_j(\varepsilon_j)=1-e^{-\kappa_j\varepsilon_j}$ with $\kappa_1=2$ and $\kappa_2=1$, $c_j(m_j)=m_j^2/2$, and $g_j(h_j,m_j)=h_je^{-m_j}$. At the harm-proportional policy $\pi^{\mathrm{hp}}=(1/2,1/2)$ and $\varepsilon^{\mathrm{hp}}=(1,1)$, the feasible privacy-budget direction $\Delta\varepsilon=(-1,1)$ strictly decreases $B_w$ locally.
\end{corollary}

\begin{proof}[Proof sketch]
At the displayed profile the lower-level cost constraint is slack, so the best response is separable. Accounting for the developer's response shows that moving privacy budget from the first statistic toward the second produces a negative directional derivative. Appendix~\ref{app:proofs} supplies the derivative and counterexample calculation.
\end{proof}

%====================================================================

\section{Optimal Private Audit Design}
\label{sec:optimal-design}

We now characterize audit designs that minimize the welfare-weighted observability gap $B_w$ when the developer responds rationally. We show three results: (i) harm-proportionality is not universally optimal, with a heterogeneous profile furnishing a counterexample (Theorem~\ref{thm:nonproportional}); (ii) the bilevel problem admits a single-level reformulation via the developer's Karush--Kuhn--Tucker (KKT) conditions (Theorem~\ref{thm:bilevel-reduction}); and (iii) we provide a projected-hypergradient heuristic, Strategic Private Audit Design (SPAD).

\subsection{When harm-proportionality fails}

\new{When harm-proportionality fails, the relevant derivative combines the audit design's direct effect on the observability weight with the developer's coupled response. No single statistic of the baseline harm profile determines the allocation:}

\begin{theorem}[Harm-proportionality is not universally optimal]
\label{thm:nonproportional}
Under Assumptions~\ref{ass:detectability}--\ref{ass:cost-mitigation}, there exist heterogeneous parameter profiles for which the joint harm-proportional allocation does not minimize $B_w(\pi,\varepsilon)$. Consequently, every minimizer $(\pi^\star,\varepsilon^\star)$ for such a profile satisfies
\[
(\pi^\star,\varepsilon^\star)\neq\left(\frac{h}{\sum_kh_k},\frac{\varepsilon_{\mathrm{tot}}h}{\sum_kh_k}\right).
\]
\end{theorem}

\begin{proof}[Proof sketch]
The gradient of $B_w$ has a \emph{direct} term in $(\pi,\varepsilon)$ through the observability-shortfall factor $1-\pi_j\alpha_j(\varepsilon_j)$ and an \emph{indirect} term through $m^*(\pi,\varepsilon)$. Implicit differentiation of the full developer KKT system captures cross-dimensional responses when the hard budget binds. The two-dimensional slack-budget profile in Corollary~\ref{cor:harm-prop}, with $\kappa_1\ne\kappa_2$, has a feasible descent direction at the harm-proportional allocation. Appendix~\ref{app:proofs} gives the coupled derivative and counterexample.
\end{proof}

\paragraph{Qualitative structure.} Welfare weights, informativeness elasticities, mitigation technology, and the hard-budget shadow price jointly shape the hypergradient. When the hard budget binds, changing one audit coordinate generally changes mitigation in every dimension. The welfare--observability mismatch parallels Stackelberg security game design \citep{Korzhyk2011-qp, Blocki2015-du, Schlenker2017-le}.

\subsection{Bilevel reduction via KKT substitution}

\new{Bilevel problems are computationally awkward because the auditor's objective contains the developer's optimizer. Under strict convexity, the developer's best response is characterized exactly by its first-order conditions, so substituting those conditions turns the two-level problem into a single-level one that a projected-gradient method can attack directly---this is what makes SPAD implementable:}

\begin{theorem}[Single-level reduction via KKT substitution]
\label{thm:bilevel-reduction}
Suppose Assumptions~\ref{ass:detectability}--\ref{ass:cost-mitigation} hold, with $c_j$ and $g_j$ twice continuously differentiable and the developer's lower-level problem strictly convex. Then the auditor's bilevel problem~\eqref{eq:auditor-problem} admits an equivalent single-level reformulation obtained by substituting the developer's KKT conditions:
\begin{align}
\min_{\pi, \varepsilon, m, \mu,\nu}\ &\sum_{j=1}^d w_j\bigl(1 - \pi_j \alpha_j(\varepsilon_j)\bigr) g_j(h_j, m_j) \label{obj:single-level}\\
\text{s.t.}\ &\pi_j \alpha_j(\varepsilon_j) \tfrac{\partial g_j}{\partial m_j} + (1 + \mu)\, c'_j(m_j)-\nu_j = 0 \quad \forall j, \label{cstr:stationarity}\\
& \sum_j c_j(m_j) \le B,\quad m_j \ge 0, \label{cstr:primal}\\
& \mu \ge 0,\quad \nu_j\ge0,\quad \mu\!\left(B - \sum_j c_j(m_j)\right) = 0,\quad \nu_jm_j=0\quad\forall j, \label{cstr:dual}\\
& \pi \in \Delta^d,\ \varepsilon \ge 0,\ \sum_j \varepsilon_j \le \varepsilon_{\mathrm{tot}}.
\end{align}
Here $\mu$ is the hard cost-budget multiplier ($=\lambda$ in~\eqref{eq:foc-dev}) and $\nu_j$ is the multiplier on $m_j\ge0$. This is a mathematical program with complementarity constraints (MPCC). Under MPEC-LICQ and strict complementarity for the active constraints---conditions that are not automatic---fixing the active set produces a local smooth nonlinear program. A slack cost budget with $\mu=0$ is nondegenerate; degeneracy instead arises, for example, when the budget binds while $\mu=0$, or when $m_j=0$ and $\nu_j=0$.
\end{theorem}

\begin{proof}[Proof sketch]
The developer's KKT conditions are necessary and sufficient based on strict convexity (Lemma~\ref{lem:uniqueness}); imposing them yields an equivalent MPCC. The MPEC-LICQ and generic strict complementarity arguments are deferred to Appendix~\ref{app:proofs}.
\end{proof}

\paragraph{Computational consequence.} On a locally fixed active set, $\nabla B_w$ admits an \emph{analytical} form from one bordered KKT solve for the indirect term plus the closed-form direct term. A \emph{numerical} variant uses $2d$ forward finite-difference perturbations through the developer's best response. Random restarts can explore different local solutions; penalty relaxation or nonlinear-program solvers such as IPOPT and SNOPT are alternatives if MPEC-LICQ is in doubt.

\subsection{Algorithm: Strategic Private Audit Design (SPAD)}

\begin{algorithm}[H]
\caption{Strategic Private Audit Design (SPAD)}
\label{alg:spad}
\small
\textbf{Input:} harm prior $h$, welfare weights $w$, informativeness functions $\{\alpha_j\}$, mitigation costs $\{c_j\}$, total privacy budget $\varepsilon_{\mathrm{tot}}$, developer type $\theta$, tolerance $\tau$, step size $\gamma$.\\
\textbf{Output:} locally improved $\pi^\star, \varepsilon^\star$.
\begin{algorithmic}[1]
\State Initialize $\pi^{(0)} \in \Delta^d$, $\varepsilon^{(0)} \in \mathcal{E}$ where $\mathcal{E} \coloneqq \{\varepsilon \in \mathbb{R}_+^d : \sum_j \varepsilon_j \le \varepsilon_{\mathrm{tot}}\}$. Set $t \gets 0$.
\While{not converged}
    \State \textbf{Inner solver:} compute $m^\star(\pi^{(t)}, \varepsilon^{(t)})$ by a convex program (projected gradient or interior point).
    \State \textbf{Objective:} $B_w^{(t)} \gets \sum_j w_j (1 - \pi^{(t)}_j \alpha_j(\varepsilon^{(t)}_j))\, g_j(h_j, m^\star_j)$.
    \State \textbf{Hypergradient:} compute $\nabla_{\pi,\varepsilon} B_w^{(t)}$ as the sum of the direct term in the observability-shortfall factor and the indirect term through $m^\star$, either via implicit differentiation of the bordered developer KKT system or via $2d$ forward finite-difference perturbations of the developer best response.
    \State \textbf{Projected step:} $x^{(t+1)} \gets \mathrm{Proj}_{\Delta^d \times \mathcal{E}}\!\left(x^{(t)} - \gamma \nabla B_w^{(t)}\right)$, where $x=(\pi,\varepsilon)$.
    \State Check the projected-step residual $\|x^{(t)}-x^{(t+1)}\| \le \tau$ (equivalently, a scaled projected-gradient mapping). $t \gets t+1$.
\EndWhile
\State \Return $(\pi^{(t)}, \varepsilon^{(t)})$.
\end{algorithmic}
\end{algorithm}

\paragraph{Complexity.} Each analytical iteration requires a bordered KKT solve; the numerical implementation instead resolves the lower-level problem after each of $2d$ finite-difference perturbations. The iteration cap and restart count are reported in Appendix~\ref{app:experiment-setup}. More advanced DP composition constraints \citep{Bassily2016-om, Lee2018-gj, Liu2019-lw, Bu2020-pl} change the auditor's feasible set.

\subsection{Robust variant: uncertainty over developer types}

When the developer's type $\theta$ (cost budget, rationality) is uncertain, we formulate
$\min_{\pi, \varepsilon} \max_{\theta \in \Theta} B_w(\pi, \varepsilon; \theta)$.
For finite or sample-approximated $\Theta$, a sample-based objective yields a tractable surrogate. Alternating auditor and adversary updates are a possible heuristic, but convergence guarantees and duality-gap bounds are deferred to future work.

%====================================================================

\section{Synthetic Evaluation Protocol}
\label{sec:experiments}

\noindent\emph{The protocol below grounds empirical follow-up work; we report a single preliminary ablation in \S\ref{sec:ablation} to ground the abstract's quantitative claim. The protocol validates SPAD's optimization \emph{given} the model's strategic-developer assumption; whether real developers behave strategically in privacy-protected audits is an empirical question for follow-up work.}

\removed{The environments span multiple dimensions, harm priors, detectability curves, cost families, privacy budgets, developer types, and auditor baselines.}\new{The protocol varies three forces: how scarce the global privacy budget is, how unevenly normalized informativeness responds across statistics, and how flexibly the developer can reallocate mitigation. Fully strategic, boundedly rational, and non-strategic developers separate equilibrium response from limited or absent response; uniform, harm-proportional, welfare-proportional, and uncertainty-focused auditors reveal which gains come specifically from anticipating that response. Because the simulation exposes both the audit-exposure score and true residual harm, it measures $B_w$ directly. Appendix~\ref{app:experiment-setup} records the full parameter grid, statistical methodology, and tables.}

\subsection{Preliminary ablation: SPAD versus na\"ive baselines}
\label{sec:ablation}

Across the reported fully strategic environments with $\varepsilon_{\mathrm{tot}}\ge0.5$, SPAD lowers $B_w$ by \spadBredMin--\spadBredMax\% relative to uniform allocation when welfare is uniform and by \spadBredNuwMin--\spadBredNuwMax\% when welfare is non-uniform. The sharper comparison is welfare-proportional allocation: SPAD's remaining \spadVsWpNuwMin--\spadVsWpNuwMax\% advantage in the non-uniform-welfare case measures the value of anticipating the developer's best response beyond simply targeting welfare. At the tightest privacy budget the tested policies are empirically similar; we do not claim a general lower-bound explanation for that pattern. These are single-environment ablations under quadratic costs and exponential residual harm, not evidence of external validity. Appendix~\ref{app:experiment-setup} reports the 8{,}000 evaluations, 50-seed uncertainty estimates, and full comparisons.

\new{\paragraph{Interpretation.} The growing SPAD advantage when welfare and observability decouple is consistent with the model's strategic mechanism, but Theorem~\ref{thm:naive-blindspot} applies only to its stated symmetric benchmark and does not explain every experimental cell. What remains untested is the central external-validity question---whether real developers best-respond to announced follow-up allocations and how production mechanisms map privacy budgets into normalized informativeness scores $\alpha_j$.}

%====================================================================

\section{Related Work}
\label{sec:related}

\paragraph{Privacy and audit-policy literatures.} Empirical DP auditors test whether implementations satisfy claimed privacy guarantees \citep{Tramer2022-zi, Nasr2023-if, Kong2024-go}. A related security literature develops attacks or defenses for measuring leakage from trained models \citep{Fredrikson2015-vq, Jia2019-ak, Nasr2019-bg, Ye2022-wc, Carlini2023-vm}, while other work studies privacy games, private federated learning, or deployment guidance \citep{Shokri2015-mj, Naseri2022-gh, Ponomareva2023-pc}. Our question differs: we use a DP-constrained release as an \emph{observability interface} for downstream-harm auditing and model a developer that reallocates mitigation in response. \citet{Yang2023-wi} give a game-theoretic analysis of auditing DP algorithms with an epistemically disparate herd; \citet{Das2026-vv} study optimal auditing of adversarial agents without DP composition.

\paragraph{Strategic agents, audit games, and privacy mechanisms.} Strategic classification \citep{Hardt2015-ci, Dong2018-ht, Chen2018-lb} and performative prediction \citep{Perdomo2020-gp} model agent gaming but not multidimensional reallocation under noised audits; audit games \citep{Blocki2015-du} and security games \citep{Korzhyk2011-qp, Schlenker2017-le, Yan2019-gc} supply the bilevel template but assume clean observations. \new{Complementary political-economy work endogenizes vendor auditability, deployer monitoring, and lock-in under evidence-dependent enforcement \citep{Burnat2026-jd}. The present model instead fixes the mechanism family and studies mitigation under per-statistic follow-up and privacy-budget allocation.} Privacy-aware mechanism design and adaptive composition \citep{Nissim2012-dp, Duchi2013-ki, Ghosh2013-wq, Bassily2016-om, Bun2018-fh, Dwork2019-vw, Bu2020-pl} treat $\varepsilon$ as a budget, not a per-statistic Stackelberg choice. Secure multiparty computation (MPC) for fairness monitoring \citep{He2026-pp} targets post-market hiring audits under a different threat model and complements our DP-noised strategic-developer regime. Our contribution is the conjunction of DP composition, per-statistic commitment, and strategic mitigation reallocation. \removed{\emph{Limitations:} one-shot commitment, exact best-response, mitigation-only adaptation; relaxations (repeated audits, noisy $\varepsilon$, concealment, bounded rationality) preserve Theorem~\ref{thm:naive-blindspot} but not SPAD optimality.}

\new{\section{Limitations}\label{sec:limitations}
The guarantees apply to a fixed, enforceable audit protocol: a regulator or trusted service executes the calibrated mechanism family, while the developer controls only mitigation. If the developer can replace or degrade the release mechanism, the feasible game expands and SPAD no longer solves it. The reduced-form $\alpha_j(\varepsilon_j)$ also freezes the residual-harm argument of the canonical power function $\widetilde\alpha_j(\varepsilon_j,g_j(h_j,m_j))$; restoring that coupling changes the developer's first-order condition and SPAD's quantitative allocation. The game is one-shot, assumes an exact developer best response and an observable announced audit design, and uses basic composition over potentially overlapping statistics. Repeated audits, bounded rationality, concealment, parallel composition for disjoint groups, and amplification from subsampling require different models. Finally, the reported ablation validates optimization within the model but does not establish that deployed developers best-respond or that production audit mechanisms follow the assumed power curves.}

\new{\section{Conclusion and future directions}\label{sec:conclusion}
Privacy-constrained audits can become strategic environments: the same observation limits that protect individuals reveal where an auditor has less statistical power. Under a fixed protocol for aggregate harm statistics, we formalized this problem as a bilevel Stackelberg game and introduced the welfare-weighted observability gap $B_w$. The effect of strategic response on $B_w$ is profile-dependent; we proved strict worsening in a symmetric two-statistic benchmark and showed by counterexample that harm-proportional allocation is not universally optimal. The full KKT-coupled hypergradient and single-level reformulation support Strategic Private Audit Design (SPAD), whose preliminary synthetic ablation lowers $B_w$ within the modeled environment. Three directions follow. \emph{Empirical:} test whether real developers respond to announced follow-up allocations and calibrate normalized power curves for production mechanisms. \emph{Theoretical:} analyze the full coupling $\widetilde\alpha_j(\varepsilon_j,g_j(h_j,m_j))$, broader sufficient conditions for strategic worsening, and repeated audits. \emph{Institutional:} endogenize the global privacy-loss envelope $\varepsilon_{\mathrm{tot}}$, which is currently fixed by the data controller, regulator, or audit agreement, as a welfare trade-off between privacy and observability \citep{Abowd2022-mf}.}

%====================================================================

\singlespacing
\printbibliography

\clearpage
\onehalfspacing
\appendix

\section{Proofs}
\label{app:proofs}

\subsection{Proof of Theorem~\ref{thm:naive-blindspot}}

At an interior optimum, the first-order conditions are
\[
\eta_i h e^{-m_i^*}=(1+\lambda)m_i^*,\qquad i\in\{1,2\}.
\]
Equivalently, $m_i^*e^{m_i^*}=\eta_i h/(1+\lambda)$. Because $m\mapsto me^m$ is strictly increasing on $\mathbb{R}_+$, $\eta_1>\eta_2$ implies $m_1^*>m_2^*$. The binding budget gives $(m_1^*)^2+(m_2^*)^2=2B$, so $m_1^*>\sqrt B>m_2^*$. Define
\[
\Delta r_i:=h\bigl(e^{-m_i^*}-e^{-\sqrt B}\bigr).
\]
Then $\Delta r_1<0<\Delta r_2$. Moreover, strict convexity of $m\mapsto e^{-m}$ and the root-mean-square inequality give
\[
\frac{e^{-m_1^*}+e^{-m_2^*}}{2}
>
e^{-(m_1^*+m_2^*)/2}
\ge
e^{-\sqrt B},
\]
where the first inequality is strict because $m_1^*\ne m_2^*$. Hence $\Delta r_1+\Delta r_2>0$. Applying identity~\eqref{eq:gap-decomp},
\begin{align*}
B_w(m^*)-B_w(m^{\mathrm{ns}})
&=u_1\Delta r_1+u_2\Delta r_2\\
&=u_1(\Delta r_1+\Delta r_2)+(u_2-u_1)\Delta r_2>0,
\end{align*}
because $u_1>0$, $u_2\ge u_1$, and $\Delta r_2>0$.

\subsection{Proof of Corollary~\ref{cor:harm-prop}}

At $\pi=(1/2,1/2)$ and $\varepsilon=(1,1)$, $0<\eta_j<1/2$. The unconstrained first-order condition is $\eta_je^{-m_j}=m_j$, so $m_j^*=W(\eta_j)<1/2$, where $W$ is the Lambert $W$ function. Thus $\sum_j(m_j^*)^2/2<1/4<B=2$, and the hard budget is slack. Define
\[
Q(\eta):=e^{-W(\eta)}\left[1+
\frac{(1-\eta)e^{-W(\eta)}}{1+W(\eta)}\right].
\]
Direct differentiation of $(1-\eta_j)e^{-W(\eta_j)}$ yields
\[
\frac{\partial B_w}{\partial\varepsilon_j}
=-\pi_j\alpha_j'(\varepsilon_j)Q(\eta_j).
\]
The function $Q$ is strictly decreasing. Indeed, $W$ is strictly increasing, so $e^{-W(\eta)}$ decreases, and $(1-\eta)e^{-W(\eta)}/(1+W(\eta))$ is a product of positive decreasing factors. At the harm-proportional policy, $\eta_1>\eta_2$, so $Q(\eta_1)<Q(\eta_2)$. In addition, $\alpha_1'(1)=2e^{-2}<e^{-1}=\alpha_2'(1)$. Therefore, along $\Delta\varepsilon=(-1,1)$,
\[
\left.\frac{d}{ds}B_w\bigl(\pi,\varepsilon+s\Delta\varepsilon\bigr)\right|_{s=0}
=\frac12\left[\alpha_1'(1)Q(\eta_1)-\alpha_2'(1)Q(\eta_2)\right]<0.
\]
The harm-proportional policy is not a local minimizer. This establishes non-universality, not failure for every heterogeneous profile.

\subsection{Coupled hypergradient and proof of Theorem~\ref{thm:nonproportional}}

Let $x=(\pi,\varepsilon)$, $\eta=\pi\odot\alpha(\varepsilon)$, and $J_\eta:=\partial\eta/\partial x$. At an interior best response with a locally fixed active set, let $J_m:=\partial m^*/\partial x$ and $J_\lambda:=\partial\lambda/\partial x$, and define
\[
H:=\operatorname{diag}\!\left(\eta\odot g''+(1+\lambda)c''\right),
\qquad q:=c'(m^*),
\]
where derivatives of $g_j$ and $c_j$ are evaluated at $m_j^*$. If the hard budget binds, implicit differentiation of the developer's stationarity and budget equations gives
\begin{equation}
\begin{bmatrix}J_m\\J_\lambda\end{bmatrix}
=-
\begin{bmatrix}H&q\\q^\top&0\end{bmatrix}^{-1}
\begin{bmatrix}\operatorname{diag}(g')J_\eta\\0\end{bmatrix}.
\label{eq:coupled-hypergradient}
\end{equation}
If the budget is slack, the final row and column are omitted and $\lambda=0$. The auditor's full hypergradient is
\begin{equation}
\nabla_x B_w
=-J_\eta^\top(w\odot g)
+J_m^\top\bigl(w\odot(1-\eta)\odot g'\bigr).
\label{eq:full-hypergradient}
\end{equation}
The bordered solve in~\eqref{eq:coupled-hypergradient} generally creates cross-dimensional terms: when the budget binds, a perturbation to one component of $x$ changes the common multiplier and hence multiple components of $m^*$. Corollary~\ref{cor:harm-prop} supplies a profile and feasible direction for which the harm-proportional allocation has a strictly negative directional derivative. It therefore proves Theorem~\ref{thm:nonproportional}.

\subsection{Proof of Theorem~\ref{thm:bilevel-reduction}}

The lower-level objective and feasible set are convex, and the objective is strictly convex under Assumption~\ref{ass:cost-mitigation}. For $B>0$, a sufficiently small strictly positive $m$ satisfies $C(m)<B$, so Slater's condition holds. The KKT conditions---stationarity~\eqref{cstr:stationarity}, primal feasibility~\eqref{cstr:primal}, dual feasibility, and both complementary-slackness conditions in~\eqref{cstr:dual}---are therefore necessary and sufficient for $m=m^*(\pi,\varepsilon)$. Substituting them into the auditor's problem yields~\eqref{obj:single-level}, which is globally equivalent to~\eqref{eq:auditor-problem}. The complementarity constraints make the reformulation an MPCC. On an active set satisfying MPEC-LICQ and strict complementarity, fixing that set gives the local smooth system used in~\eqref{eq:coupled-hypergradient} \citep{Luo2011-bi}.

\section{Concrete DP mechanism instantiations}
\label{app:dp-instantiations}

The main text abstracts the differentially private observation channel through the normalized informativeness score $\alpha_j(\varepsilon_j)$. Canonical mechanisms naturally induce a joint test-power function $\widetilde\alpha_j(\varepsilon_j,g_j(h_j,m_j))$. We obtain the reduced form by freezing the residual-harm argument at a reference action $m_j^0$, typically the prior non-strategic baseline, and rescaling power above the test's size to satisfy Assumption~\ref{ass:detectability}. The rescaled $\alpha_j$ is not itself a detection probability. Because equilibrium residual harm differs from that reference by order one in the ablations (median deviation $\approx1.07$; \texttt{scripts/verify\_local\_approx.py}), this is a modeling choice rather than a small-perturbation approximation.

\paragraph{Gaussian mechanism (approximate DP).} For an aggregate statistic $q_j(D)$ with $L_2$-sensitivity $\Delta_j$, the Gaussian mechanism adds noise $\mathcal{N}(0,\sigma_j^2I)$ with $\sigma_j=\Delta_j\sqrt{2\ln(1.25/\delta_j^{\mathrm{priv}})}/\varepsilon_j$ for $(\varepsilon_j,\delta_j^{\mathrm{priv}})$-DP. Under a one-sided $z$-test of $H_0:g_j(h_j,m_j)=0$ versus $H_1:g_j(h_j,m_j)>0$ at level $a$,
\[
\widetilde\alpha_j\bigl(\varepsilon_j,g_j(h_j,m_j)\bigr) \;=\; \Phi\!\left(\frac{\varepsilon_j g_j(h_j,m_j)}{\Delta_j\sqrt{2\ln(1.25/\delta_j^{\mathrm{priv}})}}-z_{1-a}\right).
\]
At the reference action $m_j^0$, the local reduced form $\alpha_j(\varepsilon_j):=(\widetilde\alpha_j(\varepsilon_j,g_j(h_j,m_j^0))-a)/(1-a)$ satisfies Assumption~\ref{ass:detectability}. Under R\'enyi differential privacy, one substitutes the tighter calibration from the corresponding accountant \citep{Mironov2017-df}.

\paragraph{Laplace mechanism (pure DP).} For sensitivity $\Delta_j$, the Laplace mechanism adds $\mathrm{Lap}(\Delta_j/\varepsilon_j)$. For a one-sided threshold test against null $c$ with $g_j(h_j,m_j)>c$,
\[
\widetilde\alpha_j\bigl(\varepsilon_j,g_j(h_j,m_j)\bigr) \;=\; 1-\tfrac{1}{2}\exp\!\bigl(-\varepsilon_j(g_j(h_j,m_j)-c)/\Delta_j\bigr).
\]
At reference action $m_j^0$, the rescaled $\alpha_j(\varepsilon_j):=2\widetilde\alpha_j(\varepsilon_j,g_j(h_j,m_j^0))-1$ satisfies Assumption~\ref{ass:detectability}.

\paragraph{Randomized response (local DP).} Under local DP \citep{Duchi2013-ki}, each individual answers truthfully with probability $e^{\varepsilon_j}/(e^{\varepsilon_j}+1)$. Aggregating $n$ responses, the central limit theorem gives
\[
\widetilde\alpha_j\bigl(\varepsilon_j,g_j(h_j,m_j)\bigr) \;\approx\; \Phi\!\left(\sqrt{n}\,g_j(h_j,m_j)\tanh(\varepsilon_j/2)-z_{1-a}\right),
\]
strictly increasing in both arguments. The same affine rescaling at $m_j^0$ yields a reduced form satisfying Assumption~\ref{ass:detectability}. The exponential family $\alpha_j(\varepsilon_j)=1-\exp(-\kappa_j\varepsilon_j)$ used in \S\ref{sec:experiments} approximates the Gaussian instantiation at moderate $\varepsilon_j$ after calibration to that operating point.

\paragraph{Implications for the bilevel analysis.} The reduced form $\alpha_j(\varepsilon_j)$ makes the developer's first-order condition and the auditor's hypergradient tractable. Theorem~\ref{thm:naive-blindspot} is a result about that reduced-form observability score. Restoring $\widetilde\alpha_j(\varepsilon_j,g_j(h_j,m_j))$ adds a derivative through residual harm to the developer's condition and changes SPAD's quantitative allocation. This full coupling is future work.

\section{Full experiment setup}
\label{app:experiment-setup}

This appendix instantiates the synthetic evaluation protocol of \S\ref{sec:experiments} with all design choices fixed, enabling an independent reproduction.

\paragraph{Simulation environment.} All experiments used Python 3.12 with NumPy and SciPy, using sequential least-squares programming (SLSQP) for the developer's lower-level problem. Hypergradients were computed by forward finite differences with step size $10^{-3}$ through the developer best response. Each configuration uses 50 independent seeds drawn from a fixed sequence. Code and seeds will be released at acceptance.

\paragraph{Parameterisation.} Harm-space dimensionality $d \in \{5, 10, 20\}$. Baseline harm $h$ is drawn either sparsely (component-wise $\mathrm{Beta}(0.5, 2)$) or densely (component-wise $\mathrm{Uniform}[0.5, 1.5]$), with the resulting vector normalized to $\|h\|_1 = d$. Detectability uses the exponential family $\alpha_j(\varepsilon_j) = 1 - \exp(-\kappa_j \varepsilon_j)$ with $\kappa_j \sim \mathrm{Uniform}[0.1, 2.0]$ in heterogeneous regimes and $\kappa_j = 1.0$ in homogeneous controls. Mitigation costs come from a quadratic family $c_j(m) = m^2/2$ and a power-law family $c_j(m) = m^p/p$ with $p \in \{1.5, 3\}$. Developer cost budget $B \in \{0.5d, d, 2d\}$. Total DP budget $\varepsilon_{\mathrm{tot}} \in \{0.1, 0.5, 1, 2, 5\}$, spanning tight to loose regimes. Welfare weights $w$ are drawn either uniformly ($w_j = 1/d$) or non-uniformly ($w \sim \mathrm{Dirichlet}(\alpha = 0.5)$).

\paragraph{Developer types.} \emph{Fully strategic (FS):} numerically solves the constrained lower-level program~\eqref{eq:dev-problem} with SciPy's SLSQP implementation, analytical objective gradients, tolerance $10^{-8}$, and at most 200 iterations. This is the model assumed by SPAD. \emph{Boundedly rational (BR):} runs $K=50$ projected-gradient steps from a uniform initialization $m^{(0)}=(B/d)\mathbf{1}$, with step size $\gamma=0.05$ and Armijo backtracking; this represents a developer with limited optimization capacity. \emph{Non-strategic (NS):} sets $m_j\propto w_jh_j$, ignoring the audit interface.

\paragraph{Auditor baselines.} \emph{Uniform (UNIF):} $\pi_j=1/d$, $\varepsilon_j=\varepsilon_{\mathrm{tot}}/d$. \emph{Harm-proportional (HP):} $\pi_j\propto h_j$, $\varepsilon_j\propto h_j$, normalized. \emph{Welfare-proportional (WP):} analogous with $w_j$. \emph{Uncertainty-focused (UF):} $\varepsilon_j$ proportional to the prior standard deviation of $h_j$ across seeds. \emph{Oracle (ORC):} $\varepsilon_{\mathrm{tot}}=\infty$, used only as an infeasible upper-bound reference. \emph{SPAD:} Algorithm~\ref{alg:spad} with step size $\gamma=0.1\varepsilon_{\mathrm{tot}}$ (decayed by $0.95$ every 20 outer iterations), convergence tolerance $\tau=10^{-4}$, $T_{\max}=200$, and five random restarts (best $B_w$ reported).

\paragraph{Convergence and projection.} The implementation stops when the projected-step residual is below $\tau$ or $T\ge T_{\max}$; it does not use the raw unconstrained gradient norm. Each hypergradient evaluation requires $2d$ forward finite-difference perturbations, each re-solving the developer's lower-level problem. Projection onto $\Delta^d$ uses the standard simplex projection; projection onto $\{\varepsilon\ge0,\sum_j\varepsilon_j\le\varepsilon_{\mathrm{tot}}\}$ clamps to non-negativity and rescales when the budget is violated.

\paragraph{Metrics.} The welfare-weighted observability gap $B_w(\pi, \varepsilon) = \sum_j w_j(1 - \pi_j\alpha_j(\varepsilon_j))g_j(h_j, m^*_j)$ is the auditor's objective, reported in absolute terms and as $B_w/\mathrm{TRH}$; we also track the audit-exposure score $\mathrm{DH}(\pi, \varepsilon, m^*)$ and true residual harm $\mathrm{TRH}(\pi, \varepsilon)$ as diagnostics. SPAD efficiency gain is reported as $(B_{w,\mathrm{UNIF}} - B_{w,\mathrm{SPAD}})/B_{w,\mathrm{UNIF}} \times 100\%$. Developer-type degradation is $\Delta B_w$ when SPAD is optimized assuming FS but the realized developer is BR.

\paragraph{Statistical methodology.} The preliminary ablations report means and one standard deviation across 50 seeds. The supplied aggregation scripts additionally compute a 95\% percentile-bootstrap confidence interval for the mean TRH using 10{,}000 resamples. We do not report hypothesis tests or adjusted $p$-values for this exploratory synthetic comparison.

\paragraph{Ablation design.} Six axes are varied independently with the remaining parameters fixed at the default values $(d = 10,\ \varepsilon_{\mathrm{tot}} = 1,\ \mathrm{FS\ developer},\ \mathrm{quadratic\ cost})$. \emph{A1 — Privacy budget:} $\varepsilon_{\mathrm{tot}} \in \{0.1, 0.5, 1, 2, 5\}$. \emph{A2 — Dimensionality:} $d \in \{5, 10, 20\}$. \emph{A3 — Detectability heterogeneity:} $\kappa$ uniform vs.\ spread. \emph{A4 — Developer type:} FS vs.\ BR vs.\ NS. \emph{A5 — Harm concentration:} sparse vs.\ dense $h$. \emph{A6 — Cost curvature:} $p \in \{1.5, 2, 3\}$.

\paragraph{Compute.} Each configuration runs in less than two minutes on one processor core. The full ablation requires approximately eight core-hours. No graphics processor is required.

\section{Broader Impacts}
\label{sec:broader-impacts}

Differential privacy is established in statistical publication and privacy-preserving data analysis, while its use as the observation layer for AI harm audits remains prospective. This paper studies that prospective design problem rather than documenting an existing deployment practice. Its primary positive contribution is analytical: it shows how an announced privacy-constrained observation interface can change a strategic developer's mitigation incentives. After empirical calibration and institutional validation, the framework may inform how regulators, independent auditors, or compliance teams allocate privacy and follow-up resources across harm dimensions such as demographic subgroups or failure modes.

We identify two potential negative societal impacts that warrant discussion. \emph{First, adversarial use by auditors.} An auditor with strategic interests---for example, an industry self-regulator with commercial ties to the developer---could misuse the framework to design audits that minimize regulatory exposure rather than true harm. The assumption that the welfare weights $w$ reflect genuine social priorities is critical; if $w$ is manipulated, the resulting allocation can be gamed by the auditor rather than by the developer. Possible safeguards include independent scrutiny of $w$ and disclosure of audit-design parameters.

\emph{Second, strategic complexity as a barrier to adoption.} SPAD requires a calibrated model of the developer's cost structure $(c_j, g_j)$ and mitigation budget $B$. Regulators with limited technical capacity may find the framework inaccessible relative to simpler policies. Future work should test whether simplified audit-design rules, such as threshold-based allocation, retain useful performance without requiring full developer-type knowledge.

This study did not involve human subjects, personal data, or systems deployed in production; the synthetic evaluation protocol used only simulated harm spaces.

\section{Code and data availability}
\label{app:artifact}

An anonymous reproducibility artifact is available at:
\begin{center}
\url{https://anonymous.4open.science/r/neurips-2026-strategic-dp-auditing-artifact-3952/}
\end{center}
The artifact includes the \texttt{strategic\_dp} package (model layer, developer best responses, SPAD algorithm), the experiment runner, all figures and tables in this paper, and verification scripts for the numerical claims and reduced-form approximation. The revision-verification script records the counterexample that motivated narrowing Theorem~\ref{thm:naive-blindspot}, checks the corrected direction in Corollary~\ref{cor:harm-prop}, and demonstrates the cross-dimensional response induced by a binding budget. A \texttt{Makefile} provides reproduction targets (\texttt{make help}); the headline-reproduction pipeline runs in approximately 25 minutes on one processor core, whereas the complete six-axis ablation grid described in Appendix~\ref{app:experiment-setup} requires approximately eight core-hours.

\end{document}